\documentclass[11pt]{article}

\usepackage[margin=1in]{geometry}

\usepackage{setspace}
\usepackage{amsmath, amssymb, amsthm, mathtools}
\usepackage{stmaryrd}
\usepackage{mathpartir}
\usepackage{algorithm}
\usepackage{algpseudocode}
\usepackage{wrapfig}

\usepackage{microtype}

\usepackage{graphicx}
\usepackage{booktabs}

\usepackage[colorlinks=true,allcolors=blue]{hyperref}

\newtheorem{theorem}{Theorem}
\newtheorem{lemma}[theorem]{Lemma}

\newtheorem{corollary}[theorem]{Corollary}
\theoremstyle{definition}
\newtheorem{definition}[theorem]{Definition}

\title{\Large\bfseries Static Analysis Under Non-Deterministic Program Assumptions}
\author{Abdullah H. Rasheed}
\date{}

\begin{document}

\maketitle

\begin{abstract}
Static analyses overwhelmingly trade precision for soundness and automation. For this reason, their use-cases are restricted to situations where imprecision isn't prohibitive. In this paper, we propose and specify a static analysis that accepts user-supplied program assumptions that are local to program locations. Such assumptions can be used to counteract imprecision in static analyses, enabling their use in a much wider variety of applications. These assumptions are taken by the analyzer non-deterministically, resulting in a function from sets of accepted assumptions to the resulting analysis under those assumptions. We also demonstrate the utility of such a function in two ways, both of which showcase how it can enable optimization over a search space of assumptions that is otherwise infeasible without the specified analysis.
\end{abstract}

\newpage

\newcommand{\A}[0]{\mathcal{A}}
\section{Introduction}

Static analysis is a widely used method for bug detection, compiler-level optimization, and more. This technique determines relevant properties of a given program, and it does so without executing the program. These analyses are typically categorized as ``sound over-approximations'' in that they may produce false positives. For example, a bug-detecting static analysis that is sound will report all instances of bugs, but it may be uncertain about some. These uncertainties are typically referred to as false positives. On the contrary, if the analyzer reports that no bugs are found, then there are truly no bugs.

In general, sound static analyzers must find a balance between automation, precision, and scalability. While ``assume-guarantee'' reasoning can sacrifice soundness to increase precision and scalability by allowing for local refinements, existing frameworks require assumptions to be fixed \emph{a priori}. Furthermore, this approach still treats the environment as a static input, making it brittle in the face of uncertainty. This limits the efficacy of such analyses for tasks that require optimization and synthesis, because the analyzer cannot reason about the guarantees of an assumption without first committing to it, leading to a costly trial-and-error loop. A method for analyzing under a non-deterministic choice of assumptions would bridge this gap by providing the missing ``guarantee'' information and enable optimization over sets of candidate assumptions. More formally, we propose a symbolic parameterization of static analyses that \emph{lifts} the domain $\Sigma$ to a functional domain $2^\A \to \Sigma$ where $\A$ is a set of user-supplied assumptions annotated at program locations. This immediately provides an oracle to obtain, for any desired set of assumptions, the corresponding static analysis under those assumptions. We also show that the parameterization can work on top of abstract interpreters while preserving soundness.

A number of existing works aim to solve similar goals, but our approach remains novel due to its qualification of verification by semantic, run-time assumptions. This is unlike known methods such as trace partitioning, that recover precision by partitioning the state space on conditional paths (or control-flow history in general) \cite{xavier07}. Existing lifted static analyses are highly related to our approach, but partitioning is confined to syntactic, compile-time variability (Boolean features defined at build time) \cite{chen14,parnas76}. Even approaches that support run-time variability are confined to partitioning on environment values rather than run-time program assumptions \cite{dimovski21}. Approaches like disjunctive abstract domains \cite{popeea06, sank06,xu23,giacobazzi98} obtain high precision by representing states as a disjunction of abstract elements (sometimes encoding path conditions), but they suffer from path explosion and their disjunctions are ``blind.'' They may know that the state is $A \lor B$, but they lose the reason \emph{why}, whereas our approach parameterizes the disjunction with assumptions (effectively, $(Assumption_1 \to A) \land (Assumption_2 \to B)$). This allows us to maintain a causal link between the environment, run-time assumptions, and the program state.

Unlike naive approaches that require exponential re-analysis, we perform the analysis in a single-pass while employing two strategies to obtain scalability. The first is to enforce a normal form on the parameterized states. By partitioning the assumption space into equivalence classes, we maintain a unique and compact representation that collapses redundant analysis results, enabling scalable reasoning over large sets of candidate assumptions. The second strategy can force worst-case run time down at the developer's discretion by providing a ``knob'' that balances scalability and precision. This is in the form of \emph{approximate merging}, which merges equivalence classes by joining their analysis results, resulting in over-approximation. This reduction can be applied arbitrarily, enabling \emph{controlled} approximation.

To demonstrate the utility of the parameterized analysis, we show that it simplifies an assumption synthesis problem: given a program with an assumption space (candidate assumptions) and many assert statements throughout, find a minimal set of assumptions that makes the program satisfy all assertions. The simplification of this problem also entails the simplification of many other related problems, as we discuss in the paper.

We also address potential ``inconsistencies'' in assumption sets by first showing that the parameterized analysis result enables fixpoint computation over $2^\A$. This allows us to employ a variety of known principled techniques for obtaining useful information. Particularly, we use fixpoint methods familiar to non-monotone logic to find consistent and inconsistent assumption sets.

In summary, the main contributions of this paper as follows:
\begin{itemize}
    \item \textbf{Parameterized Analyses.} We design and formalize a method for extending existing static analyses to support an ``assume'' statement, and parameterizing the analysis to yield a function from subsets of assumptions to analysis results under those assumptions. This analysis is performed in a single pass, and uses a normal form to represent the parameterized program states in a unique and compact manner.
    \item \textbf{Correctness Proofs and Abstract Interpretations.} We provide a correctness proof of the form ``for each subset $A$ of assume statements in the program, the parameterized analysis provides the same analysis result as the original analyzer under the assumptions $A$.'' We also extend this theorem to show that parameterizing a sound abstract interpreter will yield a sound analysis w.r.t. all sets of assumptions.
    \item \textbf{Sound Parameterized Approximation.} To avoid an exponential (in the number of assumptions) worst-case complexity in the parameterized analysis, we outline a method for speeding up computation by approximation. We also prove soundness of this operation and discuss approaches for losing minimal precision.
    \item \textbf{Assumption Synthesis.} To demonstrate the applicability of this analysis, we show that it can be applied to the problem of synthesizing local conditions within a program to satisfy many assert statements. This application also demonstrates where unsoundness is not exactly an issue.
    \item \textbf{Bounding Consistency in Assumptions.} We also demonstrate that the parameterized analysis can be used to prune inconsistency in assumption sets. That is, it becomes easier to find sets of assumptions that are not ``self-contradictory.'' This is because having an oracle to obtain analysis results under any given subset of assumptions allows for fixpoint computation \emph{using} the oracle. The technique draws from techniques used for non-monotone logics in knowledge representation.
\end{itemize}

\newcommand{\lfp}[1]{\operatorname{lfp}(#1)}
\section{Background}

A \emph{program} is represented as a control-flow graph (CFG) with $n$ nodes as statements and directed edges representing flow of control. An example of a written program and its corresponding CFG can be seen in Figure \ref{fig:cfg_example}.

Nodes transform the program state as defined by the language's semantics. For example, a node corresponding to the statement $x := 5$ may transform a program state $\sigma$ into $\sigma[x \mapsto 5]$ where $\sigma \in Var \hookrightarrow \mathbb{Z}$. 
\begin{wrapfigure}{r}{0.55\linewidth}
\centering
\fbox{
\begin{minipage}{0.48\linewidth}
\centering
\begin{algorithmic}
\State $x := \text{input()}$
\If{$x > 0$}
    \While{$x \leq 10$}
        \State $x := x+2$
    \EndWhile
\Else
    \State $x := 0$
\EndIf
\State \textbf{return } $x$
\end{algorithmic}
\end{minipage}
\hfill
\begin{minipage}{0.48\linewidth}
\centering
\includegraphics[width=\linewidth]{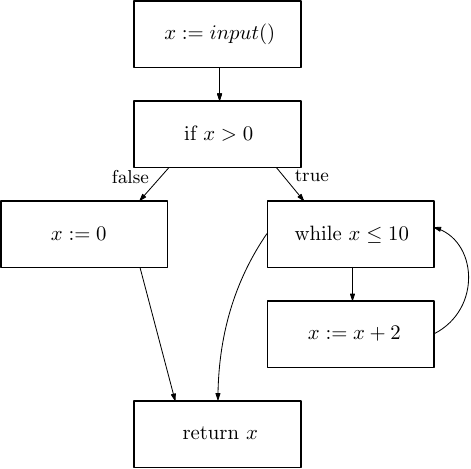}
\end{minipage}
}
\caption{Sample program and its corresponding CFG.}
\label{fig:cfg_example}

\end{wrapfigure}
The \emph{collecting semantics} describes all possible program states that can occur at each program location (CFG node). In essence, if $Concrete$ is the set of concrete program states, the collecting semantics $CS$ is defined over $\mathcal{P}(Concrete)$ such that for any CFG node $v$, $CS(v) \subseteq Concrete$ is the set of possible states that the program may hold at $v$. Statements similarly transform the collecting semantics. For example, the statement $x := 5$ may transform a set $S \subseteq Concrete$ of concrete states into the set $\{\sigma[x\mapsto 5] \mid \sigma \in S\}$ if $Concrete$ consists of maps from variables to integers.

If there is a directed edge from a CFG node $v$ to another node $v'$, we will write $v \to v'$. If the state $\sigma \in \Sigma$ can enable this transition, then we write $v \to_\sigma v'$. The set of predecessors of $v$ from a set of concrete states $S$ is defined by $$pre(v, S) = \{v' \mid \exists \sigma \in S : v' \to_\sigma v\}$$
and we will let $pre(v) = \{v' \mid \exists \sigma \in CS(v') : v'\to_\sigma v\}$. For the purposes of this paper, the definition of $pre(v)$ can be left abstract, or can even be defined by the over-approximation $\{v' \mid v' \to v\}$.

The set of possible ``incoming'' states to a CFG node $v$ is defined as $$CS_{pre}(v) = \{\sigma \in CS(v') \mid v' \in pre(v)\}$$

We will use $\Sigma$ to denote a complete lattice representing (abstract) program states. We will similarly use $\Sigma^\sharp$ to represent an abstraction of $\Sigma$ related by a Galois connection (later defined).

Throughout this paper, we will use the following system of equations to model a static analysis. Let $C_v$ be a constraint variable attached to $v$. This variable represents the (abstract) program state directly \emph{after} executing $v$ in the program. Each of these constraint variables has the following constraint $$C_v = \bigsqcup_{v' \in pre(v)} t_{v}(C_{v'})$$ for some transformer $t_v$ denoting the semantics of $v$. Observe that this is a function on all the constraint variables. That is, we may write 
\begin{align*}
    C_{v_1} &= f_{v_1}(C_{v_1},\ldots,C_{v_n})\\
    &\vdots\\
    C_{v_n} &= f_{v_n}(C_{v_1},\ldots,C_{v_n})
\end{align*}
We can further simplify this into a function $\vec f : \Sigma^n \to \Sigma^n$:
$$\vec f(\sigma_1,\ldots,\sigma_n) = (f_{v_1}(\sigma_1,\ldots,\sigma_n),\ldots,f_{v_n}(\sigma_1,\ldots,\sigma_n))$$
and the analysis result is then $\lfp{\vec f}$. As a short-hand, we will write $\llbracket P \rrbracket$ as the analysis result of a program $P$.

Let $\Sigma$ and $\Sigma^\sharp$ be complete lattices, and let $\alpha : \Sigma \to \Sigma^\sharp$ and $\gamma : \Sigma^\sharp \to \Sigma$. We say that $\alpha$ and $\gamma$ form a Galois connection if for all $\sigma \in \Sigma$ and $\sigma^\sharp \in \Sigma^\sharp$, $\alpha(\sigma) \sqsubseteq^\sharp \sigma^\sharp \Leftrightarrow \sigma \sqsubseteq \gamma(\sigma^\sharp)$. An analysis $\llbracket \cdot \rrbracket^\sharp$ over $\Sigma^\sharp$ is \emph{fixpoint sound} if for all programs $P$, $\alpha(\llbracket P \rrbracket) \sqsubseteq \llbracket P \rrbracket^\sharp$. We will simply say sound to mean fixpoint sound in this paper. We will also use $\sqsubseteq,\sqcup,\sqcap$ to be the partial order relation, join operation, and meet operation respectively of lattices, omitting additional notation indicating the lattice unless it is unclear from context. This is to keep notation less cluttered.

Let $L$ be a complete lattice. A function $f : L \to L$ is monotone if $\forall x,y\in L: x \sqsubseteq y \Rightarrow f(x) \sqsubseteq f(y)$, and it is anti-monotone if $\forall x,y \in L : x \sqsubseteq y \Rightarrow f(y) \sqsubseteq f(x)$. We will say $f^0(x)=x$, and $f^k(x)=f(f^{k-1}(x))$ when $k > 0$. The Knaster-Tarski fixpoint theorem provides that for any monotone $f$, its least fixed point $\operatorname{lfp}(f)$ and greatest fixed point $\operatorname{gfp}(f)$ necessarily exist \cite{tarski1955}. We will let $\operatorname{fix}(f)$ be the set of fixed points of any function $f : L \to L$. For a monotone $f$, we say it is $\omega$-continuous if for every $\omega$-chain $x_0\sqsubseteq x_1 \sqsubseteq \ldots$ where the supremum exists (which is always in a complete lattice), the following equality holds: $$f\left(\bigsqcup_{i \in \omega} x_i\right) = \bigsqcup_{i \in \omega} f(x_i)$$
We make the assumption that all forward transformers $f_{v_i}$ of a static analysis are $\omega$-continuous (thus requiring $\omega$-continuity in $t_{v_i}$). We later discuss how this is not a necessary requirement, but is only to simplify proofs.

Finally, by the Kleene fixed point theorem, if $f$ is $\omega$-continuous, then $\operatorname{lfp}(f) = \bigsqcup_{k\geq 0} f^k(\bot)$ where $\bot = \bigsqcup \emptyset$ (the bottom element of the lattice) \cite{stoltenberg94}.

\newcommand{\assume}[1]{\textbf{assume }#1}
\section{Program Assumptions}

We now consider extending a programming language with the statement ``$\assume{p}$.'' This statement intuitively serves to \emph{restrict} the set of executions to those that satisfy $p$, thereby enforcing $p$ on the program state. Instead of letting $p$ be a predicate over program states, we require it to represent a program state $\pi_p$. A state $\sigma$ is said to \emph{satisfy} $p$ if $\sigma \sqsubseteq \pi_p$. For example, with the collecting semantics, $p$ can be a set of concrete states. The assume statement then restricts the set of concrete states to those contained in $p$. This is, of course, unrealistic to manually write out in a program, however we will see that it is rather natural in many abstract domains.

Let $v$ be a CFG node representing an assume statement $\assume{p}$. A more formal definition of assume statement's semantics is as follows:
$$CS(v) = \{\sigma \in CS_{pre}(v) \mid \sigma \sqsubseteq \pi_p\}$$

To enforce an assumption $p$ on a program state $\sigma$, we may simply perform the meet operation $\sigma \sqcap \pi_p$. This \emph{restricts} $\sigma$ by forcing it to satisfy $p$ (since $\sigma \sqcap \pi_p \sqsubseteq \pi_p$). Observe that expressiveness in assumptions is directly related to expressiveness of the chosen domain. 

\textit{Example.} Suppose program states are in the $Interval$ domain (partial functions from variables to integer intervals). We may write assumptions of the form $\assume{x \geq 3}$ (corresponding to the state $\pi_{x\geq 3} = [x \mapsto [3,\infty]]$) or $\assume{y \leq 7}$ (corresponding to the state $\pi_{y \leq 7} = [y\mapsto [-\infty, 7]]$). We may also have logical connectives like $\assume{x \geq 3 \land x \leq 10}$ and others using any combination of $\land,\lor,\neg$. Let $\sigma = [x \mapsto [-12,8], y \mapsto [-\infty, 13]]$. To enforce $\assume{x \geq 3 \land x \leq 10}$, we may apply the meet as described before:
\begin{align*}
    \sigma \sqcap \pi_{x\geq 3 \land x \leq 10} = \sigma \sqcap [x \mapsto [3,10]] = [x\mapsto [3,8],y\mapsto [-\infty, 13]]
\end{align*}
and this map now satisfies $x \geq 3 \land x \leq 10$. 

Consider the case where $\sigma \sqcap \pi_p$ is ``infeasible.'' For example, if $p$ is $x \geq 3$ and $\sigma = [x \mapsto [-\infty, 2]]$, $\sigma \sqcap \pi_p = [x \mapsto \emptyset]$ which is not a feasible program state. In this case, we have a \emph{proof} that the assumption is false at this location and state. Recognizing infeasible states will be useful later on when searching for consistent sets of assumptions. We will let the predicate $Feasible_p(\sigma)$ denote whether $\sigma$ is feasible for $p$ (whether $\sigma \sqcap \pi_p$ is a feasible program state). We will also assume that $Feasible_p$ is monotone: if $\sigma_1 \sqsubseteq \sigma_2$ and $Feasible_p(\sigma_1)$, then $Feasible_p(\sigma_2)$.

Before defining the new analysis, we will define the transformer for an assume statement in the aforementioned analysis setup. This is for the sake of completion and a proof of correctness later on. Let $v$ be a CFG node for the statement $\assume{p}$. We define $$t_v(\sigma) = \sigma \sqcap \pi_p$$
with the semantics being that it enforces the assumption $p$ onto the input state, as described before.

For the remainder of this paper, we will let $\A$ denote the (finite) set of CFG nodes that represent assume statements in a given program. We will also denote, for all $a \in \A$, a propositional variable $\llparenthesis a \rrparenthesis$, representing whether or not $a$ is accepted as an assumption. This allows us to represent conditions on assumption sets symbolically.

\section{Parameterized Analysis}

Now we wish to perform a static analysis on a program, resulting in a function $F: 2^\A \to \Sigma^n$. That is, the result of the static analysis should be able to answer the question ``for the set $A \subseteq \A$ of assumptions, what is the static analysis result under $A$?'' (namely, $F(A)$). We would also like to avoid running a static analysis for each $A \subseteq \A$ separately. In this section, we will outline a one-pass analysis that answers this question while performing some computational optimizations. Exponential in $|\A|$ worst-case, however, is unavoidable while maintaining full precision but we will later demonstrate that sub-exponential time is achievable by sacrificing some precision.

\subsection{Parameterized Abstract Domain}

The analysis we define works by ``parameterizing'' an existing analysis. For now, we will assume the static analysis works over a single concrete domain $\Sigma$. We will later show that these methods can work on top of an abstract interpreter as well for computational feasibility. Suppose we have a sound forward transformer $\vec f : (\Sigma)^n \to (\Sigma)^n$, so that $\llbracket P \rrbracket = \operatorname{lfp}(\vec f)$.

Instead of keeping track of a single abstract state in $\Sigma$ at each location, we will keep track of a set of \emph{rules} per location. Each rule has
\begin{itemize}
    \item A \emph{condition} $\phi_i$ on assumptions. The rule is applied only for $A \subseteq \A$ that satisfy $\phi_i$.
    \item A \emph{result} $\sigma_i \in \Sigma$ corresponding to the abstract state that holds under condition $\phi_i$.
\end{itemize}
We may consider $\phi_i$ to be boolean expressions with atoms from $\A$, or equivalently as subsets of $\A$. We will consider them interchangeably for convenience, but in implementation it is more sensible to use the former representation. Formally, we are changing the abstract domain to a parameterized version $$Par(\Sigma) = \{\{\phi_i,\sigma_i\}_i \mid \phi_i \subseteq 2^\A\land \sigma_i \in \Sigma\}$$
We will strengthen this even further by requiring that, for any $\{\phi_i,\sigma_i\}_i \in Par(\Sigma)$, $\{\phi_i\}_i$ forms a disjoint partition of $2^\A$. This ensures that for all $A \subseteq \A$, exactly \emph{one} rule applies to $A$. With this restriction, we may think of $\{\phi_i\}_i$ as \emph{equivalence classes} of assumptions in the sense that if $A_1,A_2 \in \phi_i$, then under the static analysis they are locally equivalent (they yield the same analysis result $\sigma_i$ at that location).

As desired, states in $Par(\Sigma)$ represent functions in $2^\A \to \Sigma$. That is, for $X \in Par(\Sigma)$, its representative function is $$\rho\llbracket X \rrbracket(A) = \sigma_i \text{ such that } (\phi_i,\sigma_i) \in X \land A \in \phi_i$$
which is well-defined since we require that $\{\phi_i\}_i$ forms a disjoint partition of $2^\A$. Observe that lifting the canonical pointwise partial order on functions to $Par(\Sigma)$ is problematic since a single function can be represented by multiple states in $Par(\Sigma)$ (thereby failing antisymmetry). This problem will soon be addressed. 

Another problem is that $X$ may still be exponentially large. We will now present methods for reducing the size of $X$ without changing its representative function. This will not solve the worst-case, which is addressed in Section \ref{sec:approx-merging}.

\subsection{Injective Normal Form}
\label{sec:inf}

We are interested in a unique, ``most compact'' representation of $X \in Par(\Sigma)$. We will write $X \sim Y$ to mean $\rho\llbracket X \rrbracket = \rho\llbracket Y \rrbracket$ (semantic equivalence). Observe that $Par(\Sigma)/\sim$ is isomorphic to $2^\A \to \Sigma$. We now define a normal form on $Par(\Sigma)$ to assign each semantic equivalence class a unique most compact representative.

\begin{definition}[Injective Normal Form (INF)]
    $X \in Par(\Sigma)$ is in INF if, for each $(\phi_i,\sigma_i),(\phi_j,\sigma_j) \in X$ with $i \neq j$, $\sigma_i \neq \sigma_j$ and $\phi_i \neq \emptyset$.
\end{definition}

In other words, $X$ is in INF iff the function $h : X \to \Sigma$ defined by $h(\phi,\sigma) = \sigma$ is injective and $image(\rho\llbracket X \rrbracket) = image(h)$. Intuitively, INF keeps the representation compact in two ways:
\begin{enumerate}
    \item There should only be one equivalence class for each analysis result. If $\sigma_i =\sigma_j$, then $\phi_i$ and $\phi_j$ are describing conditions for the same analysis result, which is redundant.
    \item Every rule should be able to be invoked. If $\phi_i$ is unsatisfiable (considering $\phi_i$ as a predicate), its rule will never be invoked and is thus redundant.
\end{enumerate}

Given any $X \in Par(\Sigma)$, we can reduce it to INF form by repeatedly applying two operations. The first is Exact Merging, which can be performed on any two rules $(\phi_i,\sigma),(\phi_j,\sigma) \in X$ that share the same result state by merging them into one rule $(\phi_i\lor \phi_j, \sigma)$. The second operation is Redundancy Elimination, which can be performed on any rule $(\phi,\sigma) \in X$ such that $\phi$ is unsatisfiable by removing it from $X$.

\floatname{algorithm}{Example}
\textit{Example.} Consider the following program:

\begin{algorithm}[H]
    \begin{algorithmic}
\State $v_1$: $x := \text{input()}$
\State $v_2$: \textbf{assume} $x > 0$
\State $v_3$: $x := 5$
\State $v_4$: \textbf{assume} $x = 0$
\end{algorithmic}
\caption{}
\end{algorithm}
and suppose we are working over the $Interval$ domain, only tracking $x$. Also suppose that, upon entering $v_3$, the parameterized state consists of two rules: $(\llparenthesis v_2\rrparenthesis, [1,\infty])$ and $(\neg\llparenthesis v_2 \rrparenthesis, [-\infty, \infty])$. As will be seen in the following section, but we will leave as a black box for now, applying $x:= 5$ to this state yields rules $(\llparenthesis v_2\rrparenthesis, [5,5])$ and $(\neg\llparenthesis v_2 \rrparenthesis, [5, 5])$, since the value of $x$ was set to 5 regardless of the assumptions taken. In this case, we no longer need to track multiple rules, but rather we can apply Exact Merging to get the single rule $(\llparenthesis v_2 \rrparenthesis \lor \neg\llparenthesis v_2 \rrparenthesis, [5, 5])$.

To illustrate Redundant Elimination, suppose that upon entering $v_2$, the parameterized state contained the rule $(\neg \llparenthesis v_2 \rrparenthesis, [-\infty, \infty])$ (perhaps by having gone through $v_2$ in an iteration of a while loop). Then, to cover the case that we accept the assumption at $v_2$ (as will be seen in the next section, but we will again leave as a black box for now), we may transform this rule into the rule $(\neg \llparenthesis v_2 \rrparenthesis \land \llparenthesis v_2 \rrparenthesis, [1, \infty])$. Clearly, no choice of assumptions will ever satisfy $\neg \llparenthesis v_2 \rrparenthesis \land \llparenthesis v_2 \rrparenthesis$, so this rule can be eliminated.

\begin{lemma}
\label{lem:inf-existence}
    Applying Exact Merging and Redundancy Elimination to any finite $X \in Par(\Sigma)$ to a fixed point results in a state $\hat X \in Par(\Sigma)$ such that $\hat X$ is in INF and $\hat X \sim X$.
\end{lemma}
\begin{proof}
    At a fixed point, Exact Merging and Redundancy Elimination can no longer be applied. Therefore, the resulting state $\hat X$ no longer has any pair of rules that share the same result state, and it no longer has any rule whose condition is unsatisfiable. This satisfies the definition of INF, and therefore $\hat X$ is in INF.

    To show that $X \sim \hat X$, it suffices to show that a single application of Exact Merging or Redundancy Elimination preserves this equivalence.

    Suppose there exists a pair of rules $(\phi_i,\sigma),(\phi_j,\sigma) \in X$ that can be merged via Exact Merging, resulting in a new state $X'$. Then, if $A \in \phi_i\lor \phi_j$, we have $\rho\llbracket X \rrbracket(A) = \sigma = \rho\llbracket X' \rrbracket(A)$. If $A \notin \phi_i\lor \phi_j$, then the equality holds trivially, so $X \sim X'$.

    Now instead suppose that $X'$ was reached by applying Redundant Elimination on $(\phi, \sigma) \in X$. In this case, $\phi$ is empty, so $\rho\llbracket X \rrbracket(A) = \rho\llbracket X' \rrbracket(A)$ for all $A$ since $A \notin \phi$. Thus, $X \sim X'$.

    Inductively, this implies that $X \sim \hat X$.
\end{proof}

Observe that the previous lemma did not require Exact Merging or Redundancy Elimination to be performed in any particular order. In the proof of the following lemma, we also show that no matter what order these operations are performed in, they will always lead to the same state.

\begin{lemma}
    For all $X \in Par(\Sigma)$, there exists a unique $\hat X \in Par(\Sigma)$ such that $\hat X$ is in INF and $\hat X \sim X$.
\end{lemma}
\begin{proof}
    Existence is guaranteed by Lemma \ref{lem:inf-existence}. To show uniqueness, suppose there exist two states $\hat X_1,\hat X_2 \in Par(\Sigma)$ in INF such that $\hat X_1 \sim X \sim \hat X_2$ and $\hat X_1 \neq \hat X_2$. Let $\hat X_1 = \{\phi_i,\sigma_i\}_i$ and $\hat X_2 = \{\psi_i,\delta_i\}$. Since $\hat X_1 \sim \hat X_2$, we have $\rho\llbracket \hat X_1 \rrbracket(A) = \rho\llbracket \hat X_2 \rrbracket(A)$ for all $A \subseteq \A$, meaning there exists a unique (by the first condition of INF) $(\phi,\sigma) \in \hat X_1$ and $(\psi,\delta) \in \hat X_2$ such that $\sigma = \delta$ and $A \in \phi$ and $A \in \psi$. There must exist such an $A$ that $\phi \neq \psi$, since otherwise would imply that either $\hat X_1 = \hat X_2$ or $\hat X_1$, without loss of generality, has a pair $(\phi_i,\sigma_i)$ such that $\phi_i$ is unsatisfiable ($\sigma_i$ is unrealizable by the representative function). The latter possibility contradicts INF form, and the former possibility contradicts the assumption that $\hat X_1 \neq \hat X_2$. Therefore, there is such a pair of rules where $\phi \neq \psi$. This implies that there exists some $A' \in \phi$ but $A' \notin \psi$ (without loss of generality). Then, $$\rho\llbracket \hat X_1 \rrbracket (A') = \sigma = \delta \neq \rho\llbracket \hat X_2 \rrbracket(A')$$ since $\hat X_2$ is in INF, so it cannot have another rule with the result state $\delta$. This contradicts $\hat X_1 \sim \hat X_2$, thus concluding the proof.
\end{proof}

\begin{corollary}
    For all $X,Y \in Par(\Sigma)$, $X \sim Y$ iff $\hat X = \hat Y$.
\end{corollary} 

It is then evident that restricting $Par(\Sigma)$ to INF states is equivalent to $Par(\Sigma)/\sim$, with the isomorphism between $Par(\Sigma)/\sim$ and $2^\A \to \Sigma$ being defined by $\rho$. Now, we may use the canonical partial order from $2^\A \to \Sigma$ on our parameterized domain, and it becomes a complete lattice.

\subsection{Analysis Constraint Functions}

As stated earlier, we assume an existing analysis defined as the fixed point of a forward transformer $\vec f$, which can be expanded $$\vec f(\sigma_1,\ldots,\sigma_n) = (f_{v_1}(\sigma_1,\ldots,\sigma_n),\ldots,f_{v_n}(\sigma_1,\ldots,\sigma_n)).$$

To lift the forward transformer to $Par(\Sigma)^n$, we first perform a trivial lift on the constraint functions $f_{v_i}$ for all $v_i$ that are not assume statements. That is, if $v\notin \A$, then we lift it to the function $pf_{v} : Par(\Sigma)^n \to Par(\Sigma)$ by lifting $t_{v}$ to $$pt_{v}(X) = \{(\phi_i,t_{v_i}(\sigma_i)) \mid (\phi_i,\sigma_i) \in X\}$$
That is, $v_i$ transforms the result state of all given rules in the same way it would typically transform a program state. If instead $v_i$ is an assume statement, we need to modify the set of rules. For $a \in \A$, we will define $splitPair_a(\phi,\sigma)$ as the smallest function satisfying
\begin{align*}
    Sat(\phi\land\llparenthesis a \rrparenthesis) &\implies (\phi \land \llparenthesis a \rrparenthesis, \sigma \sqcap \pi_a) \in splitPair_a(\phi,\sigma)\\
    Sat(\phi\land\neg\llparenthesis a \rrparenthesis) &\implies (\phi \land \neg\llparenthesis a \rrparenthesis, \sigma) \in splitPair_a(\phi,\sigma)
\end{align*}

Intuitively, if the condition $\phi \land \llparenthesis a \rrparenthesis$ is satisfiable (realizable), then we may non-deterministically choose to enforce the assumption $a$. To do this, we update the condition to indicate that this new rule is conditional on using $a$, and we enforce the assumption on the result state $\sigma$ as discussed previously. A similar idea follows for $\phi \land \neg \llparenthesis a \rrparenthesis$, in that we indicate that the assumption is \emph{not} being taken, and the result state is kept the same (since there is no assumption being enforced). It is important to note that $\neg \llparenthesis a \rrparenthesis$ represents the fact that $a$ is not being considered in this branch, not that the assumption described in $a$ is false. Indeed, it may not be possible to enforce that $a$ is false unless $\Sigma$ is a Boolean lattice, wherein the program state can be restricted by the complement $\overline{\pi_p}$.

Looking at both cases together, if $splitPair_a(\phi,\sigma) = \{(\phi \land \llparenthesis a \rrparenthesis, \sigma \sqcap \pi_p),(\phi \land \neg\llparenthesis a \rrparenthesis, \sigma)\}$, it is clear that we are performing a split on the input rule. This represents the non-deterministic choice: either we accept the assumption or we do not. In implementation, $splitPair_a$ may check if $\phi \Rightarrow \llparenthesis a \rrparenthesis$ (and similarly for the second case). If this holds, then there is no need to conjunct $\llparenthesis a \rrparenthesis$, thereby keeping the condition smaller and less complex.

To perform splitting on a parameterized state $X \in Par(\Sigma)$, we define the natural extension $$split_a(X) = \bigcup_{(\phi,\sigma) \in X} splitPair_a(\phi,\sigma)$$
It is important to note that, if $\{\phi_i\}_i$ is a disjoint partition of $2^\A$, then $split_a$ also produces conditions that form a disjoint partition of $2^\A$. This property enables us to safely state that $split_a$ is indeed a map in $Par(\Sigma) \to Par(\Sigma)$, and this disjoint partition requirement is never violated.

\begin{lemma}
    For any $a \in \A$ and $X \in Par(\Sigma)$, let $split_a(X) = \{(\phi_i',\sigma_i')\}_i$. Then, $\{\phi_i'\}_i$ is a disjoint partition of $2^\A$, and thus $split_a(X) \in Par(\Sigma)$.
\end{lemma}
\begin{proof}
    Let $X = \{(\phi_i,\sigma_i)\}_i$. Then, $\{\phi_i\}_i$ is a disjoint partition of $2^\A$. It suffices to show that, for all $A \subseteq \A$, there exists a unique $\phi_i'$ such that $A \in \phi_i'$. We first show existence. Let $\phi_i$ be the unique condition in $\{\phi_i\}_i$ such that $A \in \phi_i$. Then w.l.o.g., if $a \in A$, $A \in (\phi_i \land \llparenthesis a \rrparenthesis)$ and this also implies that $Sat(\phi_i \land \llparenthesis a \rrparenthesis)$, so $\phi_i \land \llparenthesis a \rrparenthesis \in \{\phi_i'\}_i$, and thus we have existence. For uniqueness, observe that since $A \notin \phi_j$ for all $j \neq i$, $A \notin \phi_j \land \llparenthesis a\rrparenthesis$ and $A \notin \phi_j \land \neg\llparenthesis a\rrparenthesis$. Since this constitutes all other possible conditions in $\{\phi_i'\}_i$ (besides $\phi_i \land \neg\llparenthesis a \rrparenthesis$, which $A$ is not contained in since $a \in A$), we have uniqueness.
\end{proof}

We may now define the transformer of an assume statement. Let $v$ be the CFG node of an assume statement. We define $pt_v$ as $$pt_v(X) = split_v(X)$$ and we define the constraint function $pf_v$ for all nodes as $$pf_v(X_1,\ldots,X_n) = INF\left(\bigsqcup_{v_i \in pre(v)} pt_v(X_i)\right)$$
That is, whenever we encounter an assume statement, we perform non-deterministic case-splitting on that assumption for all input states (and compress to INF). The new system of equations becomes $$\vec{pf}(X_1,\ldots,X_n) = (pf_{v_1}(X_1,\ldots,X_n),\ldots,pf(X_1,\ldots,X_n))$$ and its least solution $\lfp{\vec{pf}} \in Par(\Sigma)^n$ represents the desired function in $2^\A \to \Sigma$ for each CFG node. We now prove correctness of this claim.

\subsection{Correctness}

Let $P$ be a program. We will write $P|_A$ for $A \subseteq \A$ to mean the program $P$ keeping only the assume nodes $A$. That is, all assume statements not in $A$ are replaced with $\assume{true}$. We will also write $\llbracket P \rrbracket_{par}$ to mean the representative function of the parameterized analysis result from parameterizing an analysis $\llbracket \cdot \rrbracket$ (as done in the previous section). The following theorem states that, for any choice of assumptions $A \subseteq \A$, the parameterized analysis correctly computes the static analysis of $P$ under the assumptions $A$.

\begin{theorem}
\label{thm:correctness}
    Given any program $P$ and any set of assumptions $A \subseteq \A$, $\llbracket P \rrbracket_{par}(A) = \llbracket P|_A \rrbracket$.
\end{theorem}
\begin{proof}
    Since $\llbracket P \rrbracket_{par} = \bigsqcup_{k\geq 0} \vec{pf}\,^k(\bot)$ and $\llbracket P|_A \rrbracket = \bigsqcup_{k\geq 0}\vec f^k_A(\bot)$ (by the $\omega$-continuity assumption) where $\vec f_A$ is the forward transformer for $P|_A$, it suffices to prove that $\vec{pf}\,^k(\bot) = \vec f^k_A(\bot)$ for all $k \geq 0$. We will do so by induction.

    The base case, $k = 0$, is trivial. Let $k \geq 1$. Then, by the inductive hypothesis we have $\rho\llbracket \vec{pf}\,^{k-1}(\bot)\rrbracket (A) = \vec f^{k-1}_A(\bot)$. Let $\vec f^{k-1}_A(\bot) = (\sigma_1,\ldots,\sigma_n)$ and $\vec{pf}\,^{k-1}(\bot) = (X_1,\ldots,X_n)$, and let
    \begin{align*}
        \vec{pf}\,^{k}(\bot) &= (X_1',\ldots,X_n') = (pf_{v_1}(X_1,\ldots,X_n),\ldots,pf_{v_n}(X_1,\ldots,X_n))\\
        \vec f^k_A(\bot) &= (\sigma_1',\ldots,\sigma_n') = (f_{v_1}(\sigma_1,\ldots,\sigma_n),\ldots,f_{v_n}(\sigma_1,\ldots,\sigma_n))
    \end{align*}
    We will show that $\rho\llbracket X_i'\rrbracket(A) = \sigma_i'$ for all $i \in [n]$. There are two cases to consider. The first case is that $v_i$ is not an assume statement ($v_i \notin \A$). Here, $$\sigma_i' = \bigsqcup_{v_j \in pre(v_i)} t_{v_i}(\sigma_j) = \bigsqcup_{v_j \in pre(v_i)} t_{v_i}(\rho\llbracket X_j\rrbracket (A))$$ by the inductive hypothesis, and $$X_i' \sim\bigsqcup_{v_j \in pre(v_i)} pt_{v_i}(X_j)$$ (by Lemma \ref{lem:inf-existence}) so $$\rho\llbracket X_i'\rrbracket(A) = \bigsqcup_{v_j \in pre(v_i)} \rho\llbracket pt_{v_i}(X_j)\rrbracket (A)= \bigsqcup_{v_j \in pre(v_i)} t_{v_i}(\rho\llbracket X_j\rrbracket(A))$$
    by isomorphism of $\rho$ and definition of $pt_{v_i}$ for $v_i \notin \A$. Thus, $\rho\llbracket X_i'\rrbracket(A) = \sigma_i'$ in this case.

    The second case is that $v_i \in \A$. For all $j$, let $(\phi,\sigma_j) \in X_j$ be the unique rule such that $A \in \phi$ (the result state is $\sigma_j$ since $\rho\llbracket X_j\rrbracket(A) = \sigma_j$). Then, if $v_i \in A$, $$\sigma_i' = \bigsqcup_{v_j \in pre(v_i)}\sigma_j \sqcap \pi_{p}$$ and $$\rho\llbracket X_i'\rrbracket(A) = \bigsqcup_{v_j\in pre(v_i)} \rho\llbracket split_{v_i}(X_j)\rrbracket(A) = \bigsqcup_{v_j\in pre(v_i)} \sigma_j \sqcap \pi_p$$
    If $v_i \notin A$, we similarly have
    \begin{align*}
        \sigma_i' &= \bigsqcup_{v_j \in pre(v_i)}\sigma_j\\
        \rho\llbracket X_i'\rrbracket(A) &= \bigsqcup_{v_j\in pre(v_i)} \sigma_j
    \end{align*}
    In both cases, $\rho\llbracket X_i'\rrbracket(A) = \sigma_i'$.
\end{proof}

\paragraph{Remark.} The $\omega$-continuity requirement on the forward transformers helps in the proof where the analysis is exact w.r.t. the domain being analyzed over, and is not a restricting requirement. It seems restricting, however, when analyses desire to use techniques like widening that do not compute the least fixed point of the forward transformer. While we do not completely formalize it, these cases do not actually break the theorem as seen in the proof, since the parameterized analysis and the original analysis under $A$ are equivalent in ``lock-step'' regardless of the restrictions put on the forward transformer. While in these cases they will not compute the least fixed point as assumed by the theorem, the analysis results will be the same so long as they are iterative.

\subsubsection{Soundness with Abstract Interpreters}

In practice, we will want to parameterize the abstract domain rather than the concrete one. Of course, we may simply replace $\Sigma$ with the abstract domain $\Sigma^\sharp$. We ought to ensure that the resulting parameterized analysis result is sound w.r.t. the parameterized analysis result under the concrete domain, where $\alpha : \Sigma \to \Sigma^\sharp$ is the abstraction function of the corresponding Galois connection.

Let $\llbracket \cdot \rrbracket^\sharp$ be the result of a sound abstract interpretation and let $\llbracket \cdot \rrbracket_{par}^\sharp$ be the analysis result from parameterizing that sound abstract interpretation.

\begin{theorem}
    Given any program $P$ and any set of assumptions $A \subseteq \A$, $\alpha(\llbracket P|_A\rrbracket) \sqsubseteq \llbracket P \rrbracket_{par}^\sharp(A)$
\end{theorem}

\begin{proof}
    By Theorem \ref{thm:correctness}, we have $\llbracket P \rrbracket_{par}^\sharp(A) = \llbracket P|_A\rrbracket^\sharp$. Then, since $\llbracket \cdot \rrbracket^\sharp$ is sound w.r.t. $\llbracket \cdot \rrbracket$, we have $$\alpha(\llbracket P|_A\rrbracket) \sqsubseteq \llbracket P|_A\rrbracket^\sharp = \llbracket P \rrbracket_{par}^\sharp(A)$$
\end{proof}

\subsection{Approximate Merging}
\label{sec:approx-merging}
As frequently mentioned in prior sections, we cannot avoid exponentiality while maintaining precision. Therefore, we will forego some precision in exchange for sub-exponential state representations.

In Section \ref{sec:inf}, we discussed methods for reducing a parameterized state $X \in Par(\Sigma)$ to an equivalent and more compact form (INF). We will introduce a new reduction operation similar to Exact Merging, except that it is applicable to any pair of rules.

\paragraph{Approximate Merging.} Given two rules $(\phi_i,\sigma_i),(\phi_j,\sigma_j) \in X$, we can merge them into a single rule $(\phi_i \lor \phi_j, \sigma_i \sqcup \sigma_j)$. This rule intuitively states ``If $A$ satisfies $\phi_i$ \emph{or} $\phi_j$, then its resulting state is \emph{either} $\sigma_i$ \emph{or} $\sigma_j$.'' Of course, some of this information is lost since the joined state cannot recover the two original states.

The result of applying Approximate Merging is an over-approximation of the initial state.

\begin{theorem}
    Let $X \in Par(\Sigma)$ be non-empty. Then, applying Approximate Merging to any two rules in $X$ results in a state $X'$ such that $X \sqsubseteq X'$.
\end{theorem}
\begin{proof}
    Suppose we apply Approximate Merging to $(\phi_i,\sigma_i),(\phi_j,\sigma_j) \in X$, resulting in $(\phi_i \lor \phi_j, \sigma_i \sqcup \sigma_j)$. Then, for all $A \subseteq \A$, if $A \in \phi_i \lor \phi_j$, it is immediate that $$\rho\llbracket X \rrbracket (A) \sqsubseteq \sigma_i \sqcup \sigma_j = \rho\llbracket X' \rrbracket (A)$$
    On the other hand, if $A \notin \phi_i \lor \phi_j$, then $\rho\llbracket X \rrbracket (A) = \rho\llbracket X' \rrbracket (A)$.
\end{proof}

One may decide to apply Approximate Merging arbitrarily, but it is better to do it smartly. Ideally, we want to lose as little precision as possible. Suppose we have some function $Loss(\sigma_i,\sigma_j)$ that allows us to measure the amount of precision lost between those input states and their join $\sigma_i \sqcup\sigma_j$. Ideally, when performing Approximate Merging we would like to choose $(\phi_i,\sigma_i),(\phi_j,\sigma_j)$ that minimizes $Loss(\sigma_i,\sigma_j)$. Future works may consider designing data structures to efficiently support this for certain classes of domains by taking advantage of some structure.

\paragraph{Additional Remarks.} 
The use of Approximate Merging can be compared to widening. Widening sacrifices precision for termination and speed, while Approximate Merging sacrifices for precision (but termination is not affected when $\A$ is finite). Fortunately, in most applications, the number of assumptions is small enough and INF reduction does enough work to keep representation computationally reasonable. Approximate Merging, in most cases, is then a reduction that only needs to be applied conservatively at some threshold determined by the developer. One should also be careful about pre-emptively applying Approximate Merging, since using it too early can cause a loss of precision that could have been avoided in later iterations by Exact Merging. For instance, suppose $X$ contains far too many rules, so Approximate Merging is applied. It could have been the case that, had Approximate Merging not been applied, nearly all rules in $X$ would attain the same result state in one more iteration, thereby enabling a massive Exact Merging. It could then be said that the Approximate Merging was done in vain. This is comparable to the consequences of applying widening too early.

\newcommand{\assert}[1]{\textbf{assert } #1}
\section{Assumption Synthesis}

While a static analysis under non-deterministic program assumptions can be used in many scenarios, we look at an example use-case in this section. Consider a program with assert statements throughout. Our goal is to find a set of assumptions $A \subseteq \A$ such that each assertion is satisfied, if such a set exists.

The solution to this problem is rather straightforward after performing the analysis and obtaining $\llbracket P \rrbracket_{par}$. Let $v$ be the CFG node of an assert statement $\assert{\varphi_v}$, and $X_v$ be the analysis result at $v$. All that needs to be done is to search for a rule $(\phi_i,\sigma_i) \in X_v$ such that $\sigma_i \models \varphi$. Then, for all $A \in \phi_i$, taking the assumptions $A$ makes the program satisfy the assertion at that location. Generally speaking, $\phi_i$ describes a class of such assumption sets. To find a set of assumptions (if one exists) that satisfies \emph{all} assertions, we find all sets of assumptions that satisfy each individual assertion, and then find a common assumption set among them all. More formally, the sets of assumptions that provably (w.r.t. the analysis result) satisfy all assertions is $$Syn(P) = \bigcap_{\text{assertions } v}\bigcup_{\substack{(\phi_i,\sigma_i) \in X_v \\ \sigma_i \models \varphi_v}} \phi_i$$

The reason that the solution to this problem becomes so simple is that we have the guarantees for each assumption set before having to fix any assumptions. Assumptions no longer need to be ``derived.'' Instead, we have all the necessary information to begin with (of course, working over a finite set of user-supplied candidate assumptions also simplifies the problem greatly). With sound over-approximations (such as those produced by abstract interpretation or liberal use of Approximate Merging), however, the information provided by the analysis may be too coarse to prove or disprove the assertions. In this case, we may state ``unknown.'' Here, we maintain soundness in that whenever a set of assumptions is returned, then it truly proves the assertions true. Otherwise, we may output ``unknown'' (or if we can disprove the assertions for all assumption sets, then we may output ``impossible''). 

The following theorem formalizes this by stating that any set of assumptions returned indeed makes the assertion true under the semantics of the original analysis $\llbracket \cdot \rrbracket$. The following theorem is conditional on $Syn(P)$ being non-empty.

\begin{theorem}
    Let $A \in Syn(P)$ and $\llbracket P|_A \rrbracket = (\sigma_{v_1},\ldots,\sigma_{v_n})$. Then, for all CFG nodes $v_i$ representing assertions $\assert{\varphi_{v_i}}$, $\sigma_{v_i} \models \varphi_{v_i}$.
\end{theorem}
\begin{proof}
    Let $v$ be the CFG node of an assert statement $\assert{\varphi_v}$. Then, since $A \in Syn(P)$, we have $A \in \cup \{\phi_i \mid (\phi_i,\sigma_i) \in X_v \land \sigma_i \models \varphi_v\}$ where $X_v$ is the analysis result at $v$ in $\llbracket P\rrbracket_{par}$. Therefore, $\exists (\phi_i,\sigma_i) \in X$ such that $\sigma_i \models \varphi_v$ and $A \in \phi_i$. Since $A \in \phi_i$, we have $\rho\llbracket X_v \rrbracket(A) = \sigma_i$, so $\rho\llbracket X_v \rrbracket(A) \models \phi_i$. By Theorem \ref{thm:correctness}, $\rho\llbracket X_v\rrbracket(A) = \sigma_v$, so $\sigma_v \models \varphi_v$.
\end{proof}

Consider the classic problem that asks ``find a pre-condition on a given program $P$ such that a given post-condition $\varphi$ is satisfied.'' This is, generally speaking, a weaker version of the problem previously discussed. While the two problems can be encoded in the other's form, the classic problem requires disjunctions and other additional complexities to encode path-dependencies, and cannot recover path information. On the other hand, our discussed problem naturally handles path-dependencies and path information since assume/assert statements can be placed anywhere in the program, even inside conditionals. Therefore, assume/assert statements inherently have path information baked in and do not need disjunctions to represent statements that the classic problem would need. This demonstrates that our parameterized analysis is applicable to a wider class of problems in general, as many problems can be encoded in the form of our described assumption synthesis problem without gaining additional complexity.

\paragraph{Remarks.}
There is a difference in expressibility between assumptions and assertions. While assumptions are restricted to program states in the analysis domain (the abstract domain), assertions can be anything \emph{provable} by program states. For example, the $Interval$ domain cannot express the assumption $x \leq y$, but it can \emph{prove} the assertion $x \leq y$ in a program state where $[x \mapsto [a,b], y\mapsto [c,d]]$ and $b \leq c$.

Observe that these proofs can only be done with enough precision in the analysis, and we obtain increased precision by strengthening program states with assumptions. Therefore, precision and proof ability is influenced by two things: the assumptions and the chosen abstract domain. The assumptions can enable enough precision to prove \emph{stronger} properties, while the abstract domain can enable enough expressiveness to prove more \emph{kinds} of properties. 

As a final note, it is important to note what is being said in this problem: let $A \in Syn(P)$. Then, \emph{if} the assumptions $A$ are taken to be true, \emph{then} all assertions are true. This is a conditional statement, where the validity of the assumptions in $A$ is not necessarily guaranteed. Because the problem is framed this way, unsoundness in the static analysis does not cause issues. If the assumptions contradict the true reachable program states, however, then knowledge of this can assist developers discover key issues in a large code base or even prove mutual exclusion properties of given assumptions. In the following section, we present a principled framework for eliminating such contradictions.

\section{Consistent Assumption Sets}

As stated before, unsoundness is not always an issue since the parameterized analysis result is of the form ``\emph{if} the assumptions $A$ are taken to be true, \emph{then} the program state at $v$ is $\sigma$.'' Sometimes, however, it is useful to eliminate \emph{inconsistent} sets of assumptions. In this section, we show that the parameterized analysis provides an easy means to do so.

We say that a set of assumptions $A \subseteq \A$ is \emph{consistent} if the static analysis under $A$ cannot disprove any assumptions in $A$ nor can it allow any new assumptions in $\A$. Let $\Phi(A)$ be the set of assumptions in $\A$ that are not refuted by the static analysis under $A$. A consistent set is then a fixed point of $\Phi$. More formally, for a statement $\assume{p}$, let $$\zeta_p(X) = \{\phi \mid (\phi,\sigma) \in X \land \neg Feasible_p(\sigma)\}$$
This is the set of conditions in $X$ that yield infeasible program states w.r.t. the assumption $p$. The set of assumption sets that \emph{refute} an assumption $a \in \A$ is then $$B_a(X_{v_1},\ldots,X_{v_n}) = \bigcup_{\phi_i \in \zeta_p(X_a)} \phi_i$$
where $a$ is the CFG node for a statement $\assume{p}$. Finally, $\Phi$ can be formally defined as $$\Phi(A) = \{a \in \A \mid A \notin B_a(\llbracket P|_A \rrbracket)\}$$
That is, $\Phi(A)$ is the set of assumptions that are not refuted by $A$. Taking a closer look at this function, it is evident that $\Phi$ is anti-monotone. This is because larger $A$ yields a more precise analysis (more precise $X_a$ for all $a \in \A$), which makes it easier to refute more assumptions (monotonicity of $Feasible$), which leaves fewer remaining assumptions. The dual statement can also be made.

\newcommand{\fix}[1]{\operatorname{fix}(#1)}
\newcommand{\gfp}[1]{\operatorname{gfp}(#1)}
Since $\Phi$ is anti-monotone, $\Phi^2$ is monotone and thus has a least fixed point by the Knaster-Tarski fixed point theorem. Furthermore, it is known that $\lfp{\Phi^2} \subseteq \fix{\Phi}$ and $\fix{\Phi} \subseteq \Phi(\lfp{\Phi^2})$ (and that $\Phi(\lfp{\Phi^2}) = \gfp{\Phi^2}$). This provides upper and lower bounds on the fixed points of $\Phi$ (the consistent assumption sets). Let $\mu = \lfp{\Phi^2}$ and $\nu = \gfp{\Phi^2}$. These bounds allow us to make certain inferences, such as:
\begin{itemize}
    \itemsep-1em 
    \item If $a \in \mu$, then $a$ appears in \emph{every} consistent set of assumptions.\\
    \item If $a \in \nu$, then $a$ appears in \emph{some} consistent set of assumptions.
    \begin{itemize}
        \item If $a \notin \nu$, then $a$ \emph{never} appears in any consistent set of assumptions.
    \end{itemize}
\end{itemize}

This delivers a principled method for guiding tools towards consistent assumptions. Observe that this fixpoint computation only works because we have the resulting analysis function in $2^\A \to \Sigma$. Without this, computing $\Phi(A)$ would require re-running the analysis under $A$ which is often infeasible (since $\Phi$ needs to be iterated potentially many times to reach $\lfp{\Phi^2}$).

Also note that $\mu$ and $\nu$ may form a very poor bound. In fact, $(\mu,\nu)$ may span the entire lattice, even if the fixed points are consolidated. There are existing studies on efficiently obtaining tighter bounds \cite{anon26}, which can make this consistency search more feasible and useful.

\section{Related Works}

Lifted static analysis solves a similar problem of analyzing a program parameterized by environments. In the case of lifted static analysis, whose main application targets Software Product Lines (SPLs), analyses are lifted to support syntactic, compile-time variability as static configuration options ``\texttt{ifdef FEATURE}'' \cite{dimovski19,chen14,dimovski22}. In \cite{dimovski21}, configuration options are relaxed to be dynamic by treating features as first-class program variables. In both cases, these works operate on \emph{features}, while our work operates on \emph{assumptions}. There is a key difference, that features in lifted static analyses are typically global or scoped variables intended to be ``configurations'' while assumptions in our work are \emph{semantic constraints} introduced at arbitrary program points. These are not ``variable features'' designed by a programmer, but are rather formal constraints treated non-deterministically by the analyzer. Furthermore, consistency in our approach is dependent on semantics rather than fixed ``feature models.'' Our approximate merging also enables sound merging of assumption paths, providing a precision/scalability knob. These differences represent a significant transfer of technology from software engineering to formal verification.

Trace partitioning is a well studied approach to mitigate precision loss by distinguishing execution traces based on the history of control flow decisions \cite{xavier07}. Our approach extends this philosophy of partitioning from control flow to the assumption space. While trace partitioning splits the state based on conditionals (like \texttt{if (c) ...}), our approach splits state based on assumptions (like \texttt{assume (P) ...}). This partitioning makes the analysis state conditional not just on the program's internal history, but on additional ``hypothetical behavior.'' This is a shift from refining \emph{control} to refining \emph{context}.

Assumption synthesis in assume-guarantee (AG) reasoning has seen various approaches, each with their own merits. One well-known example is learning-based AG automation \cite{chaki07,blundell05,Pasareanu2008,nam06}, which learns assumptions as regular languages in a ``guess-and-check'' loop between a learner and teacher. Despite known optimizations, this learning process is computationally expensive and primarily effective for finite-state properties. In our approach, we eliminate the iteration loop common in these learning algorithms (most often $L^*$ learning) and instead perform \emph{symbolic assumption propagation}. By parameterizing the analysis of a component (as AG works to verify properties in components of a system) under non-deterministic assumptions, our approach computes the \emph{entire function} of the component's behavior relative to the assumption space. The task of proving a property of the whole system then moves from a model checking task to a consistency check between the parameterized results of the various components (by domain intersection, satisfiability, or other means). This positions our method as a ``white-box'' approach, as opposed to the ``black-box'' learning of standard techniques. It leverages the static analysis of the source code to derive the dependency structure between assumptions and guarantees explicitly.

The application of our analysis to assumption synthesis is adjacent to existing abductive inference-guided methods \cite{spies24,albarghouthi16,dillig12}. While our approach has a less forgiving precision/scalability trade-off than these approaches, it maintains a \emph{family} of sufficient conditions rather than single weakest conditions. By tracking how conditional dependencies evolve at every program point, we effectively maintain a ``running weakest pre-condition'' parameterized by abstract semantic assumptions rather than raw logic. This also demonstrates a difference between our approach and other model checking guided approaches \cite{wang05,giannakopoulou05}. This difference also outlines the trade-offs between the kinds of approaches.

\section{Conclusion}

We defined an analysis over a parameterized abstract domain equivalent to functions from sets of program assumptions to corresponding analysis results under those assumptions. These assumptions were provided at program locations by the user, and treated non-deterministically by the analyzer. To keep the parameterized state representation computationally feasible, we enforced a unique compact representation on the parameterized domain. This also gave it the properties of a complete lattice. We provided a method for approximating parameterized states in exchange for improved scalability, enabling a precision/scalability knob at the developer's discretion. To illustrate the usefulness of the parameterized analysis, we demonstrated how it simplifies assumption synthesis problems by changing the problem into an optimization over a search space that no longer needs to be discovered. We also showed how the analysis result being functional enables useful fixpoint computation, which we took advantage of to guide a search for consistent sets of assumptions.

\bibliographystyle{plain}
\bibliography{refs}

\end{document}